\documentclass[12pt]{article}
\usepackage{tikz}
\usepackage[english]{babel}
\usepackage{graphicx}
\usepackage{amsmath}
\usepackage{amsthm}
\usepackage{amssymb}
\usepackage{lscape}
\usepackage{caption}
\usepackage[hyperindex,plainpages=false]{hyperref}
\usepackage[left=1.1in,right=1.1in,bottom=1in,top=1in,a4paper]{geometry}

\newtheorem{theorem}{Theorem}
\theoremstyle{definition}
\theoremstyle{lemma}
\newtheorem{definition}{Definition}

\newtheorem{lemma}{Lemma}

\title{Infinitely many conservation laws for the discrete KdV equation}
\author{Alexander G. Rasin and Jeremy Schiff 
\\ Department of Mathematics,
\\ Bar-Ilan University,
\\ Ramat Gan, 52900, Israel
\\ {rasin@math.biu.ac.il},~~~{schiff@math.biu.ac.il }}

\begin{document}
\maketitle
\begin{abstract}
In \cite{RH3} Rasin and Hydon suggested a way to construct an infinite number of conservation 
laws for the discrete KdV equation (dKdV), by repeated application of a certain 
symmetry to a known conservation law. It was not decided, however, whether the resulting 
conservation laws were distinct and nontrivial. In this paper we obtain the following 
results: (1) We give an alternative
method to construct an infinite number of conservation laws using a discrete version of 
the Gardner transformation. 
(2) We give a direct proof that
the Rasin-Hydon conservation laws are indeed distinct and nontrivial. (3) We consider
a continuum limit in which the dKdV equation becomes a first-order eikonal equation. In 
this limit the two sets of conservation laws become the same, and are evidently
distinct and nontrivial. This proves the nontriviality of the conservation laws 
constructed by the Gardner method, and gives an alternate proof of 
the nontriviality of the conservation laws 
constructed by the Rasin-Hydon method. 
\end{abstract}

\section{Introduction}

The theory of integrable quad-graph equations starts from the works of Hirota \cite{Hi1,Hi3,HiSa}, where the author presents a nonlinear partial difference equation (P$\Delta$E) which ``reduces to the Korteweg-de Vries equation in the weakly nonlinear and continuum limit'' \cite{Hi1}. More recently many interesting properties of quad-graph equations have been found that can be interpreted as integrability properties \cite{Za0}. Lax pairs for some quad-graph equations were presented in \cite{BS0,NC0}. The ultra-local singularity confinement criterion for quad-graphs equations was developed in \cite{GKT,GRP,RGH}. This criterion can be viewed as a discrete version of the Painlev\'e property. Further research showed that certain quad-graph equations are consistent on the cube and all equations satisfying this condition were classified \cite{ABS,BS1}. From this consistency one can derive the Lax pair and B\"acklund transformations. So consistency on the cube also becomes an integrability criterion, even though there is no analog of this in the theory of partial differential equations (PDE). The derivation of mastersymmetries for certain quad-graph equations can be found in \cite{RH2}. The mastersymmetry can be used to construct infinite hierarchies of symmetries, and this can also can be interpreted as an integrability property.

One more criterion for integrability is the existence of an infinite number of conservation laws (in involution, in the case of a Hamiltonian system). The investigation of conservation laws of quad-graph equations was initiated by Hydon in \cite{Hy0}. There the author presented a method for computation of conservation laws of quad-graph equations and derived conservation laws for the modified discrete Korteweg-de Vries (mdKdV) equation. Hydon's method was improved in \cite{RH0,RH1,RH3}; in these papers the authors also derived three and five point conservation laws for all the equations in the ABS classification. In \cite{RH3} a suggestion for constructing an infinite number of conservation laws was given (for many of the equations in the ABS classification). But it was not shown that all the resulting conservation laws were distinct and nontrivial. 

In this paper we focus on conservation laws for the discrete KdV equation 
(dKdV, $\mathbf{H1}$ in the ABS classification) 
\begin{equation}(u_{0,0}-u_{1,1} ) ( u_{1,0}-u_{0,1} ) +\beta-\alpha=0.\label{eq01}
\end{equation}
Here $k,l\in \mathbb{Z}$ are independent variables
and $u_{0,0}=u(k,l)$ is a dependent variable that is defined on the
domain $\mathbb{Z}^2$. We denote the values of this variable at
other points by $u_{i,j}=u(k+i,l+j)=S_k^iS_l^ju_{0,0}$, 
where $S_k,~S_l$ are the unit forward shift operators in $k$ and $l$ respectively. 
In \cite{RH3} it was shown that by applying the symmetry 
\[X={\frac{k}{u_{1,0}-u_{-1,0}}}\frac{\partial}{\partial{u_{0,0}}}
   -\partial_\alpha\]
to the dKdV conservation law 
$$
F=-\ln\left(u_{0,1}-u_{-1,0}\right),~~~G=\ln\left(u_{1,0}-u_{-1,0}\right), 
\label{basicsym}
$$
and then adding a trivial conservation law, we obtain a new nontrivial dKdV 
conservation law 
\[ F_{new}=\frac{-1}{(u_{0,0}-u_{-2,0})(u_{0,1}-u_{-1,0})},~~~
 G_{new}=\frac1{(u_{0,0}-u_{-2,0})(u_{1,0}-u_{-1,0})}.\]
It was suggested that an infinite number of conservation laws could be generated by
repeating this procedure.

The first result of this paper is an alternative 
method to construct an infinite number of conservation laws using a discrete version of 
the {\it Gardner transformation}. The Gardner transformation is an
elementary method to construct the infinite number of conservation laws of 
the continuum KdV equation \cite{DrJo}. We belive the conservation laws of dKdV 
obtained from this new method are the same as those obtained from the 
Rasin-Hydon symmetry approach, but do not prove it. Our second result 
is a direct proof of the nontriviality of the Rasin-Hydon conservation laws. 
The proof exploits a fundamental lemma about conservation laws of a particular form as
well as properties of the discrete Euler operator. Our third contribution is to consider
a certain continuum limit of dKdV, in which the equation becomes a first-order eikonal 
equation. In this limit the two sets of conservation laws become the same, and are evidently
nontrivial. This proves the nontriviality of the conservation laws constructed by the 
Gardner method, and gives an alternate proof of 
nontriviality of the conservation laws constructed by the Rasin-Hydon method. 
It also provides evidence for our hypothesis that the conservation 
laws constructed by the two methods do indeed coincide.  

The structure of this paper is as follows: Section 2 is a summary of the general 
theory of conservation laws, including a new lemma about a particularly significant 
kind of conservation laws for P$\Delta$Es. Section 3 presents the Gardner method 
for dKdV. Section 4 gives the proof of nontriviality of the Rasin-Hydon conservation 
laws. Section 5 decribes the continuum limit. Finally, section 6 contains some 
concluding comments and questions for further study. 

\section{Conservation laws} 

We find it useful to summarize the standard results on conservation laws for both 
PDEs and P$\Delta$Es, as general background to the paper, and in particular as
background for a crucial lemma we will need for  section 4. 

For a scalar partial differential equation 
with two independent variables $x,t$ and a single
dependent variable $u$, a (local) conservation law is an expression of the 
form 
$$ \partial_t G + \partial_x F  = 0 \label{claw} $$
which holds as a consequence of the equation. 
Here $F,G$,  which are called ``the components of the conservation law'',
are functions of $x,t,u$ and a finite number of  partial derivatives of $u$. 
For example, if $u$ satisfies  the KdV equation 
$$ u_t = \frac14 u_{xxx} + 3 uu_x  $$ 
we then have 
\begin{eqnarray*} 
\partial_t \left( u \right) + 
\partial_x \left( -\frac14u_{xx} - \frac32 u^2 \right)  &=& 0 \ ,\\
\partial_t \left( u^2 \right) + 
\partial_x \left( -\frac12 uu_{xx} + \frac14 u_x^2 - 2u^3 \right)  &=& 0 \ ,\\
\partial_t \left( 4u^3 - u_x^2  \right) + 
\partial_x \left( -9u^4 + \frac12 u_xu_{xxx} - \frac14 u_{xx}^2 
        -3u^2u_{xx} + 6uu_x^2  \right)  &=& 0 \ .
\end{eqnarray*}
We say a conservation law is trivial for an equation if by application of the 
equation to the individual components of the conservation law we can bring them
into a form for which the law holds for all functions $u$, not just on solutions of 
the equation. Equivalently, we say the conservation law with components 
$F,G$ is trivial if we can write 
\begin{eqnarray*} 
F &=& F_0 - \partial_t f  \\
G &=& G_0 + \partial_x f  
\end{eqnarray*}
where $F_0,G_0$ both vanish as a consequence of the equation and $f$ is an 
arbitrary function of 
$x,t,u$ and a finite number of  partial derivatives of $u$ \cite{Ol0}. 
For KdV, and more generally for any equation of the form $u_t=p(x,t,u,u_x,u_{xx},\ldots)$,
there is a simple way to recognize nontrivial conservation laws. We first 
use the equation to eliminate all occurences of $t$ derivatives in $G$. 
The conservation law is trivial if and only if the resulting component $G$ is 
a (total) $x$-derivative (of some function $f$ of 
$x,t,u$ and a finite number of $x$-derivatives of $u$). Below we will 
prove an analog of this result for dKdV. 

Moving now to the discrete case, we consider a general quad-graph equation 
\begin{equation} 
P(k,l,u_{0,0},u_{1,0},u_{0,1},u_{1,1},{\bf a})=0 
\label{eq2o}\end{equation}
where ${\bf a}$ denotes a vector of parameters. 
A conservation law is an expression of the form
\begin{equation}
(S_l-I)G+(S_k-I)F=0
\label{eq5o}
\end{equation}
which holds as a consequence of the equation. 
$F$ and $G$ are called ``the components of the conservation law'',
and are functions of $k,l$, the parameters ${\bf a}$, and the values of
the variable $u$ at a finite number of points. In the above $I$ denotes 
the identity mapping. 

We say a conservation law is trivial for an equation if it takes the form 
\begin{eqnarray} 
F &=& F_0 - (S_l-I) f  \nonumber \\
G &=& G_0 + (S_k-I) f  \label{trivcl1} 
\end{eqnarray}
where $F_0,G_0$ both vanish as a consequence of the equation and $f$ is an 
arbitrary function of $k,l,{\bf a}$ and the values of
the variable $u$ at a finite number of points. 

The conservation laws we will study in this paper for dKdV have a special form: 
\begin{definition}
We say the conservation law (\ref{eq5o}) for the quad-graph equation (\ref{eq2o})
is ``on the horizontal line'' if $G$ depends only on values of $u_{n,m}$ with $m=0$,
and ``on the vertical line''  if $F$ depends only on values of $u_{n,m}$ with $n=0$.
\end{definition}
We then have the following lemma: 
\begin{lemma}\label{lemma1}
For the dKdV equation, a trivial conservation law on the horizontal (vertical) line 
can always be presented in a form with $G=(S_k-I)f$ ($F=(S_l-I)f$), where $f$ 
is a function  on the horizontal (vertical) line.
\end{lemma}
This is a discrete analog of the usual way to recognize trivial conservation laws 
for the continuum KdV equation, as described above. Even though we state the lemma
here just for the dKdV equation, it in fact holds for other quad-graph equations; 
the equation is only used in the proof in a weak way.
Note that if the $G$ of the trivial conservation law on the 
horizontal line in the lemma depends on $u_{n,0}$ with $n_l\le n\le n_r$, then $f$ 
depends on $u_{n,0}$ with $n_l\le n\le n_r-1$ (and similarly in the vertical case). 
We present the proof just for the horizontal case, the vertical case is similar. 
\begin{proof}
Since the conservation law is trivial, 
we can write $G=G_0+(S_k-I)f$, where $G_0$ vanishes on solutions of the equation, 
but it is possible that $G_0$ and $f$ may not depend only on values 
on the line, see figure 1. 
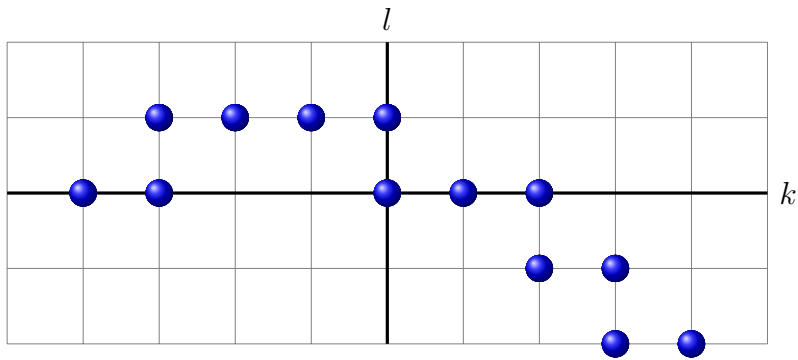
\begin{figure}[ht]
\begin{center}
\begin{tikzpicture}
\draw[step=1cm,gray,very thin] (0,0) grid (10,4);
\draw[very thick] (0,2)--(10,2)node[right]{$k$} (5,0)--(5,4)node[above]{$l$};
\foreach \x/\y in {1/2,2/2,2/3,3/3,4/3,5/3,5/2,6/2,7/2,7/1,8/1,8/0,9/0}
{\shade[ball color=blue] (\x,\y)circle(0.18cm);} 
\end{tikzpicture}
\end{center}
\caption{Possible points of $G_0$}
\end{figure}
The forms of $G_0$ and $f$ are not uniquely determined; we have the freedom 
to add to $f$ any term that vanishes on solutions of the equation and subtract 
a corresponding term from $G_0$. Using the equation in the form 
\begin{equation}
u_{1,1}=\omega_2(u_{0,0},u_{0,1},u_{1,0})\ ,\label{d161}
\end{equation} 
we can eliminate the occurence in $f$ of all values of $u$ above the horizontal line
(where $G$ is valued), with the possible exception of values on a vertical line going 
through the leftmost point on the horizontal line. Similarly, using the 
equation in the form 
\begin{equation}
u_{1,0}=\omega_1(u_{0,0},u_{0,1},u_{1,1})\label{d162}
\end{equation} 
we can eliminate in $f$ values of $u$ below the horizontal line and to the right 
of the vertical line. Thus without loss of generality 
we can take $f$ to be valued on a ``tetris'', see figure 2.
\begin{figure}[ht]
\begin{center}
\begin{tikzpicture}
\draw[step=1cm,gray,very thin] (0,0) grid (10,4);
\draw[very thick] (0,2)--(10,2)node[right]{$k$} (5,0)--(5,4)node[above]{$l$};
\foreach \x/\y in {3/3,3/4,3/1,3/0,3/2,4/2,5/2,6/2,7/2,8/2}
{\shade[ball color=blue] (\x,\y)circle(0.18cm);} 
\end{tikzpicture}
\end{center}
\caption{Points of $f$ (after moving to a tetris)} 
\end{figure}
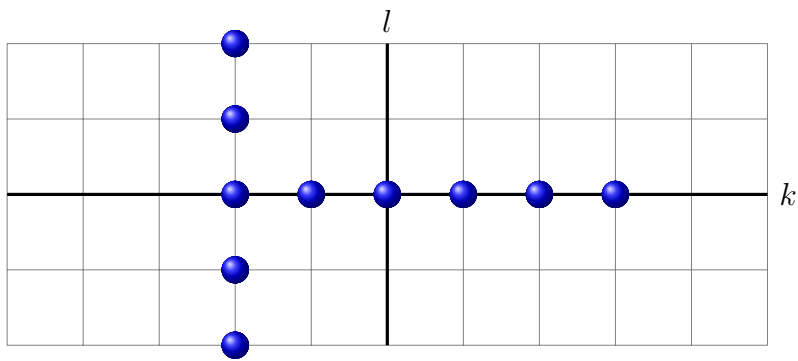
More formally, we have 
$$
f = f\left(u_{n_l,m_b},u_{n_l,m_b+1},\ldots,u_{n_l,m_t}, u_{n_1+1,0},u_{n_l+2,0},\ldots,u_{n_r,0}
    \right)\ . 
$$
Here $m_b\le 0$ and $m_t\ge 0$ denote the lowest and highest values on the vertical axis,
and $n_l$, $n_r$ the lowest and highest on the horizontal axis. The corresponding 
points of $(S_k-I)f$ are indicated in figure 3. 
\begin{figure}[ht]
\begin{center}
\begin{tikzpicture}
\draw[step=1cm,gray,very thin] (0,0) grid (10,4);
\draw[very thick] (0,2)--(10,2)node[right]{$k$} (5,0)--(5,4)node[above]{$l$};
\foreach \x/\y in {4/3,4/4,4/1,4/0,3/3,3/4,3/1,3/0,3/2,4/2,5/2,6/2,7/2,8/2,9/2}
{\shade[ball color=blue] (\x,\y)circle(0.18cm);} 
\end{tikzpicture}
\end{center}
\caption{Points of $(S_k-I)f$ (after moving to a tetris)}  
\end{figure}
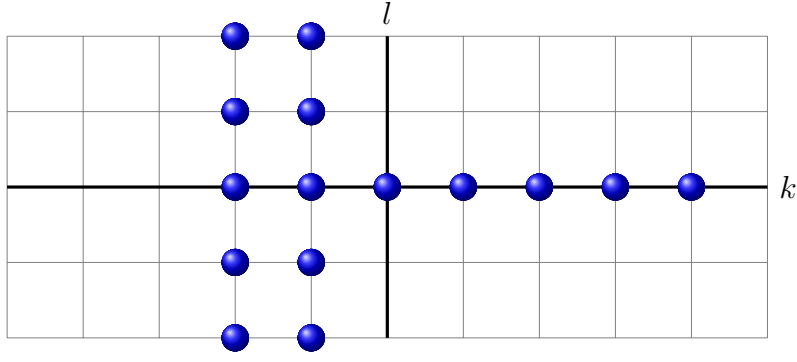
We know that $G=G_0+(S_k-I)f$ where $G$ is valued on the horizontal line 
and $(S_k-I)f$ is valued on the points shown in figure 3. Assume now that $m_t>0$,
i.e. that $f$ depends nontrivially on a point above the line. We write 
\begin{eqnarray*}
&& (S_k-I)f \\
&=& 
f\left(u_{n_l+1,m_b},u_{n_l+1,m_b+1},\ldots,u_{n_l+1,m_t}, u_{n_1+2,0},u_{n_l+3,0},\ldots,u_{n_r+1,0}
 \right)  \\
&-&  
f\left(u_{n_l,m_b},u_{n_l,m_b+1},\ldots,u_{n_l,m_t}, u_{n_1+1,0},u_{n_l+2,0},\ldots,u_{n_r,0}  
 \right) \\
&=& \left[
f\left(u_{n_l+1,m_b},u_{n_l+1,m_b+1},\ldots,u_{n_l+1,m_t}, u_{n_1+2,0},u_{n_l+3,0},\ldots,u_{n_r+1,0}
 \right)  \right. \\
&-&\left.  
f\left(u_{n_l+1,m_b},u_{n_l+1,m_b+1},\ldots,\omega_2(u_{n_l,m_t-1},u_{n_l,m_t},u_{n_1+1,m_t-1}), 
      u_{n_1+2,0},u_{n_l+3,0},\ldots,u_{n_r+1,0}
 \right)  \right]\\
&+&  \left[
f\left(u_{n_l+1,m_b},u_{n_l+1,m_b+1},\ldots,\omega_2(u_{n_l,m_t-1},u_{n_l,m_t},u_{n_1+1,m_t-1}), 
      u_{n_1+2,0},u_{n_l+3,0},\ldots,u_{n_r+1,0}
 \right)  \right. \\
&-& \left. 
f\left(u_{n_l,m_b},u_{n_l,m_b+1},\ldots,u_{n_l,m_t}, u_{n_1+1,0},u_{n_l+2,0},\ldots,u_{n_r,0} 
  \right) \right]
\ .
\end{eqnarray*}
The term in the first square brackets here 
vanishes on solutions of the equation. 
The term in the second square brackets only depends on a single value at the vertical 
level $m_t$, $u_{n_l,m_t}$. Evidently, if we want $G=G_0+(S_k-I)f$ to hold, where $G$ is
valued on the line and $G_0$ vanishes on solutions of the equation, we must demand
that the term in the second square brackets in fact be independent of $u_{n_l,m_t}$
on solutions of the equation,
i.e. that 
\begin{eqnarray*}
&& \left.\frac{\partial f}{\partial u_{n_l,m_t}}
\right|_{\left(u_{n_l+1,m_b},u_{n_l+1,m_b+1},\ldots,\omega_2,u_{n_1+2,0},u_{n_l+3,0},\ldots,u_{n_r+1,0}\right)} 
\left.\frac{\partial \omega_2}{\partial u_{0,1}}\right|_{(u_{n_l,m_t-1},u_{n_l,m_t},u_{n_1+1,m_t-1})}  \\
&=& 
\left.\frac{\partial f}{\partial u_{n_l,m_t}}
\right|_{\left(u_{n_l,m_b},u_{n_l,m_b+1},\ldots,u_{n_l,m_t}, u_{n_1+1,0},u_{n_l+2,0},\ldots,u_{n_r,0}\right)} 
\end{eqnarray*}
on solutions of the equation. (On the LHS of the above formula $\omega_2$ is written 
as short for  $\omega_2(u_{n_l,m_t-1},u_{n_l,m_t},u_{n_1+1,m_t-1})$.) 
Clearly the necessary identity  will hold if $\frac{\partial f}{\partial u_{n_l,m_t}}=0$ i.e. 
if $f$ is independent of $u_{n_l,m_t}$. Further manipulation (using the equation 
to make the two sides of the last equation depend on the same 
values of $u$) shows that indeed this is the only case. 
Thus we have a contradiction, our assumption that $f$ depends nontrivially on 
a point above the line has been proved wrong. 

By a similar argument  we show that 
$f$ does not depend on points below the line. Finally, if $G=G_0+(S_k-I)f$ and 
both $G$ and $f$ are on the line, and $G_0$ vanishes on solutions of the 
equation, clearly we must have $G_0=0$ and  the lemma is established.  
\end{proof}

Having established this lemma (at least for the case of the dKdV equation, and it is 
also true in some generality), it just remains to indicate how this allows us to 
easily identify nontrivial conservations laws on the line. For this
we use the discrete Euler operator. The discrete Euler operator is defined 
\cite{HM0} by 
\begin{equation}
E(A)=\sum_{n,m}S_k^{-n}S_l^{-m}\left(\frac{\partial A}{\partial u_{n,m}}\right).
\end{equation}
Here $A$ is a function of finitely many values $u_{n,m}$ of the variable $u$. Clearly
$$ E(S_k A) = E(S_l A) = E(A)\ , $$
and thus 
$$ E((S_k-I)A) = E((S_l-I)A) = 0\ . $$ 
So given the $G$ ($F$) component of a conservation law on the horizontal (vertical) line,
if application of the Euler operator does not give zero, it is nontrivial.  

\section{The Gardner method for dKdV Conservation laws} 

Before presenting the Gardner method for generating conservation laws of dKdV,
we review the method for continuum KdV. 
The Gardner method for KdV starts with the B\"acklund transformation.  The
B\"acklund transformation states that if $u$ solves KdV then so does $u+v_x$ 
where $v$ is a solution of the system 
\begin{eqnarray*}
v_x &=& \theta - 2u - v^2  \\
v_t &=& -\frac12 u_{xx} + \left(\theta+u\right)\left(\theta - 2u - v^2\right) + u_x v
\end{eqnarray*}
It is straightforward to check that if $u$ solves KdV then these two equations 
for $v$ are consistent (i.e. $(v_x)_t=(v_t)_x$) and also that if $v$ is defined by these
two equations then $u+v_x$ does indeed satisfy KdV. Here $\theta$ is a parameter. 
The next thing to do is to observe that if we could solve the first equation of the 
B\"acklund transformation to write $v$ as a function of $u$
and (a finite number of) its  $x$-derivatives, then we would have the following 
conservation law:   
$$ 
\partial_t v + \partial_x \left( \frac12 u_x - (u+\theta)v  \right)   
=0\ . $$
This cannot be done explicitly, but it is possible to write a formal 
solution of the first equation of the B\"acklund transformation to give $v$ in
terms of $u$ as a formal series in decreasing powers of $\theta^{1/2}$. The
first few terms of the relevant series are 
$$ 
v = \theta^{1/2}
- \frac{u}{\theta^{1/2}}
+ \frac{u_x}{2 \theta}
- \frac{u_{xx}+2u^2}{4\theta^{3/2}}
+ \frac{u_{xxx}+8uu_x}{8\theta^{2}}
- \frac{u_{xxxx}+8u^3+10u_x^2+12uu_{xx}}{16\theta^{5/2}}
+ O\left(\theta^{-3}\right)\ .
$$ 
Each coefficient in this expansion gives (the $G$ component of) 
a conservation law. More precisely, the 
coefficients of integer powers of $\theta$ give trivial conservation laws, and the 
coefficients of half integer powers give the ``$G$'' components of nontrivial 
conservation laws. The ``$F$'' components can be found  from the corresponding
coefficient in the expansion of $\frac12 u_x - (u+\theta)v$. Examining in detail the 
way in which the terms of the above series are generated, it can be shown
that the coefficient of $\theta^{-n+1/2}$ has a term proportional to $u^n$, for 
$n=1,2,3,\ldots$ and thus the corresponding conservation law is  nontrivial 
\cite{DrJo}.

We now try to reproduce this for dKdV. 
The B\"acklund transformation for dKdV \cite{AHN,Je0}
is $u\rightarrow \tilde{u}$ where
\begin{eqnarray*}
(\tilde{u}_{0,0} - u_{0,1})(u_{0,0}-\tilde{u}_{0,1}) &=& \theta-\beta \ , \\
(\tilde{u}_{0,0} - u_{1,0})(u_{0,0}-\tilde{u}_{1,0}) &=& \theta-\alpha \ .
\end{eqnarray*}
Here $\theta$ is a parameter. Once again, by a B\"acklund transformation 
we mean two things: that the above equations for $\tilde{u}$ are 
consistent if $u$ satisfies dKdV, and that $\tilde{u}$ defined by these 
equations also satisfies dKdV. 
As in the case of continuum KdV, we cannot in general solve the equations of the 
B\"acklund transformation to write $\tilde{u}$ in terms of $u$.
However there are  several special cases. In the 
case $\theta=\beta$ we can take $\tilde{u}_{0,0} = u_{0,1}~{\rm or}~u_{0,-1}$, 
and in the 
case $\theta=\alpha$ we can take $\tilde{u}_{0,0} = u_{1,0}~{\rm or}~u_{-1,0}$. 
For $\theta$ near these special values we can find series solutions.
Consider the case $\theta=\alpha+\epsilon$ where $\epsilon$ is small, and look for
a solution of the B\"acklund transformation in the form 
$$ \tilde{u}_{0,0} = u_{1,0} + \sum_{i=1}^\infty v^{(i)}_{0,0} \epsilon^i \ . $$
We just look at the second equation of the B\"acklund transformation. This reads
$$ \epsilon = 
  \left(  \sum_{i=1}^\infty v^{(i)}_{0,0} \epsilon^i  \right) 
  \left(  u_{0,0}- u_{2,0} - \sum_{i=1}^\infty v^{(i)}_{1,0} \epsilon^i  \right) 
\ . $$
The leading order approximation gives 
\begin{equation} 
v^{(1)}_{0,0} = \frac1{u_{0,0}-u_{2,0}} \ .  \label{v1} 
\end{equation}
Higher order terms give 
\begin{equation}  v^{(i)}_{0,0} = \frac1{u_{0,0}-u_{2,0}} 
   \sum_{j=1}^{i-1} v^{(j)}_{0,0} v^{(i-j)}_{1,0} 
  \ , \qquad i=2,3,\ldots \ . \label{v2} 
\end{equation}
We note that all these formulas are on the horizontal line, i.e. all the $v^{(i)}_{0,0}$ 
only depend on values of $u_{n,m}$ with $m=0$. 
More precisely, $v^{(i)}_{0,0}$ depends only on 
$u_{n,0}$ with $0\le n\le (i+1)$, is homogeneous of degree $1-2i$ in these variables, and 
only depends on these variables through the combinations 
$u_{2,0}-u_{0,0},u_{3,0}-u_{1,0},\ldots,u_{i+1,0}-u_{i-1,0}$. 

As in the case of continuum KdV an infinite sequence of conservation laws can 
be obtained starting from the $\epsilon$ expansion of  a single ``conservation law''
written in terms of $u$ and $\tilde{u}$. It is straightforward to check that if we define 
$$  F=-\ln\left(\tilde{u}_{0,0} - u_{0,1} \right) \ , \quad 
    G=\ln\left(\tilde{u}_{0,0} - u_{1,0}  \right)   
$$     
then 
$$ (S_l-I)G  + (S_k-I) F = 
    \ln
\frac{\left(\tilde{u}_{0,0} - u_{0,1} \right)\left(\tilde{u}_{0,1} - u_{1,1}  \right)   }
     {\left(\tilde{u}_{1,0} - u_{1,1} \right)\left(\tilde{u}_{0,0} - u_{1,0}  \right)   }  
= 0 \ .
$$

It only remains to explicitly expand $F$ and $G$ in powers of $\epsilon$. We have 
$$ F= -\ln\left(u_{1,0} - u_{0,1} + \sum_{i=1}^\infty v^{(i)}_{0,0} \epsilon^i  \right) 
 = -\ln\left(u_{1,0} - u_{0,1}\right) - \ln\left( 1 + 
     \frac1{u_{1,0} - u_{0,1}}\sum_{i=1}^\infty v^{(i)}_{0,0} \epsilon^i  \right) \ , $$
\begin{equation}
G= \ln\left( \sum_{i=1}^\infty v^{(i)}_{0,0} \epsilon^i  \right) 
 = \ln\epsilon - \ln\left(u_{0,0}-u_{2,0}\right)  + \ln\left( 1 + 
   \frac1{v^{(1)}_{0,0}}\sum_{i=1}^\infty v^{(i+1)}_{0,0} \epsilon^i  \right)\ . 
\label{G} \end{equation}
Writing $F=\sum_{i=0}^\infty F_i\epsilon^i$, $G=\ln\epsilon+\sum_{i=0}^\infty G_i\epsilon^i$ 
and introducing the notation  
$$A_i=S_k^i\left(\frac1{u_{0,0}-u_{2,0}}\right)\ , \quad  i=0,1,2,\ldots \ ,   
  \qquad B=\frac1{u_{1,0}-u_{0,1}}\ , $$
we obtain 
\begin{eqnarray}
&&
\left\{ 
\begin{array}{l}
F_0 = \ln B  \\ 
G_0 = \ln A_0 
\end{array}
\right. \ ,   \nonumber \\
&&
\left\{ 
\begin{array}{l}
F_1 = -BA_0  \\  
G_1 = A_0A_1  
\end{array}
\right. \ ,   \nonumber \\
&&
\left\{ 
\begin{array}{l}
F_2 =  -A_0^2A_1B + \frac12 A_0^2B^2  \\ 
G_2 =  A_0A_1^2A_2 + \frac12 A_0^2A_1^2
\end{array}
\right. \ ,  \nonumber \\
&&\left\{
\begin{array}{l}
F_3 =  -A_0A_1^2A_2^2B -  A_0^3A_1^2B + A_0^3A_1B^2 - \frac13A_0^3B^3   \\
G_3 =  A_0A_1^2A_2^2A_3 + A_0^2A_1^3A_2 + A_0A_1^3A_2^2 + \frac13A_0^3A_1^3 \\  
\end{array}\right. \quad {\rm etc.}  \label{FsGs} 
\end{eqnarray}
For $i>0$, $F_i$ is homogeneous of degree $2i$ in the variables $A_0,A_1,\ldots,A_{i-1},B$
and $G_i$ is homogeneous of degree $2i$ in the variables $A_0,A_1,\ldots,A_{i}$. 
We note that in all the conservation laws that we have computed, the $G$ components
are a sum of terms of the form ``$p_kp_{k+1}$''. Thus, for example, in $G_3$ there 
are terms of the form $A_0^2A_1^3A_2$  (take $p_k=A_0^2A_1$)  and 
$A_0A_1^3A_2^2$ (take $p_k=A_0^1A_1^2$) but no terms of the form 
$A_0^3A_1^2A_2$ or $A_0^2A_1^2A_2^2$. 

Thus we see how expansion of the B\"acklund transformation around the point
$\theta=\alpha$ yields an infinite sequence of conservation laws on the horizontal
line for dKdV. Expansion around $\theta=\beta$ yields an infinite sequence of
conservation laws on the vertical line. We have not yet given a proof that the 
conservation laws we have found by the Gardner method are all nontrivial, 
but this will emerge from analysis of the continuum limit in section 5. 

\section{The symmetry method for dKdV Conservation laws} 

We now turn to the symmetry method for generating dKdV conservation laws, 
as proposed by Rasin and Hydon \cite{RH3}.
The method proceeds by the repeated application of a certain symmetry 
to a certain basic conservation law.

We start by reveiwing the necessary theory \cite{Hy1,RH1}.   
An infinitesimal symmetry for the quad-graph  equation (\ref{eq2o})
is an infinitesimal transformation of the form 
\begin{eqnarray} 
u_{0,0} &\rightarrow& \hat{u}_{0,0} = u_{0,0}+\epsilon Q(k,l,u,{\bf a})
          +O(\epsilon^2)\ ,  \label{TT1} \\
{\bf a} &\rightarrow& \hat{\bf a} = {\bf a}+\epsilon \xi({\bf a}) 
          +O(\epsilon^2)\ ,  \nonumber 
\end{eqnarray} 
that maps solutions into solutions. 
The functions $Q$ and $\xi$  are called the characteristics of the symmetry;
the function $Q$ depends on finitely many shifts of $u_{0,0}$.
The symmetry is often written in the form 
\begin{equation}
X=Q\frac{\partial}{\partial u_{0,0}}  + \xi\cdot\frac{\partial}{\partial {\bf a}} \ ,
\label{X}\end{equation}
which is also referred to as the symmetry generator. 
By shifting (\ref{TT1}) in the $k$ and $l$ directions we obtain that under the 
symmetry
\begin{equation*}
u_{i,j} ~\rightarrow~ 
\hat{u}_{i,j}=u_{i,j}+\epsilon S_k^iS_l^jQ+O(\epsilon^2),
\end{equation*}
for every $i,j\in\mathbb{Z}$. Expanding (\ref{eq2o}) to first
order in $\epsilon$ yields the symmetry condition
\begin{equation*}
 \hat{X}(P)=0\qquad\text{whenever (\ref{eq2o}) holds,}
\end{equation*}
where $\hat{X}$ is the ``prolonged''  generator:
\begin{equation}
\hat{X}=\sum_{i,j}S_k^{i}S_l^{j}(Q)\frac{\partial}{\partial u_{i,j}} 
+ \xi\cdot\frac{\partial}{\partial {\bf a}} \ . 
\label{XX}\end{equation}
In the sequel we will not distinguish between the generator of 
a symmetry and its prolongation,  typically
from the context it is clear which one is meant.  

$X_m$ is a {\em mastersymmetry} \cite{Fe0,Fu0,Ra0,SK0} for the symmetry $X$ if it
satisfies
\begin{equation}
\left[X_m,X\right]\neq0,\qquad\left[\left[X_m,X\right],X\right]=0 \ .
\end{equation}
Here $\left[\cdot ,\cdot\right]$ denotes the commutator. 

Given a conservation law, a new conservation law can be obtained by applying a 
symmetry generator.
This is possible since the (prolonged) symmetry generator commutes with shift operators, i.e.
\[
[{X},S_k]=0,~~~~[{X},S_l]=0\ .
\]
So if $F$ and $G$ are components of a conservation law, i.e. 
\[
(S_k-I)F+(S_l-I)G=0|_{P=0}\ , 
\]
then by applying the symmetry generator we obtain
\[
0={X}((S_k-I)F)+{X}((S_l-I)G)|_{P=0}=(S_k-I){X}(F)+(S_l-I){X}(G)|_{P=0}\ . 
\]
Thus  
\[F_{new}={X}(F),~~~~G_{new}={X}(G),\]
is also a conservation law.

Known symmetry generators for dKdV (\ref{eq01}) include 
\begin{align}
&X_0=\frac{1}{u_{1,0}-u_{-1,0}}\frac{\partial}{\partial u_{0,0}}, 
&\quad&
Y_0=\frac{1}{u_{0,1}-u_{0,-1}}\frac{\partial}{\partial u_{0,0}},
&\nonumber\\
&X=\frac{k}{u_{1,0}-u_{-1,0}}\frac{\partial}{\partial u_{0,0}}-\partial_{\alpha},
&\quad&
Y=\frac{l}{u_{0,1}-u_{0,-1}}\frac{\partial}{\partial u_{0,0}}-\partial_{\beta}.
&\label{SH1}
\end{align}
(The rationale for the notation here should become clear later.) 
Known conservation laws include 
\begin{eqnarray}
&& F=\ln(u_{0,1}-u_{-1,0})\ ,~G=-\ln(u_{1,0}-u_{-1,0})\ ,\nonumber\\
&& \bar{F}=\ln(u_{0,1}-u_{0,-1})\ ,~ \bar{G}=-\ln(u_{1,0}-u_{0,-1})\ ,\nonumber\\
&& \tilde{F}=kF+l\bar{F}\ ,~ \tilde{G}=kG+l\bar{G}\ .\label{CLH1}
\end{eqnarray}
(The first of these is the leading order conservation law found by the 
Gardner method in the previous section. The second is the leading order 
conservation law found using the Gardner method expanding the B\"acklund 
transformation around $\theta=\beta$.)

The Rasin-Hydon method for constructing an infinite sequence of conservation
laws is as follows: 

\begin{theorem}\label{theorem1}
The dKdV equation has an infinite number of nontrivial conservation laws on the 
horizontal line, generated by repeated application of the symmetry $X$ to the 
conservation law with components $(F,G)$, and  
an infinite number on the vertical line, generated by repeated application of the 
symmetry $Y$ to the conservation law with components $(\bar{F},\bar{G})$. 
\end{theorem}

\begin{proof}
We look at the horizontal case, the vertical case  is similar. 

It is known \cite{RH1} that equation (\ref{eq01}) has  infinite hierarchies of symmetries 
in both the $k$ and $l$ directions. The symmetries in the $k$ direction can be obtained by 
repeatedly commuting $X$ with $X_0$:
\[
X_{1}=\left[X,X_0\right],~~~X_{2}=\left[X,X_{1}\right],
~~\ldots \ .\]
(Similarly the symmetries in the $l$ direction are obtained by commuting $Y$
with $Y_0$.) 
The generator of $X_n$ is $Q_n\frac{\partial}{\partial u_{0,0}}$ where the 
characteristic $Q_n$ depends (at most) on the $2n+3$ variables
$u_{-n-1,0},u_{-n,0},\ldots,u_{n,0},u_{n+1,0}$. 
In particular, note that $Q_n$ does not 
depend on $k$;  $X$ does, but since $X$ is a mastersymmetry for $X_0$, 
the $k$-dependence disappears on forming the 
necessary commutators. All the ${X}_{n}$ are different (linearly independent).  

Let us denote 
\[  F_n=X^{n}(F)\ , \quad G_n=X^{n}(G)\ , \quad n=0,1,2,\ldots\ . \]
From the forms of $X$ and $G$ it follows that $G_n$ depends (at most) on $k$ and the 
$2n+3$ variables
$u_{-n-1,0},u_{-n,0},\ldots,u_{n,0},u_{n+1,0}$. 
It is homogeneous of order  $-2n$ in the $u$ variables. From this homogeneity 
it follows that so long as the $G_n$ are nontrivial 
then they will also not be dependent. Furthermore the $G_n$ are valued on the 
horizontal line. By the lemma of section 2 it follows that if the conservation
law $G_n$ is trivial we must $G_n=(S_k-I)f_n$ for some function $f_n$ valued on 
the line, and therefore $E(G_n)=0$, where $E$ is the Euler operator. Thus we
can prove nontriviality by verifying that $E(G_n)\not=0$. We now show that  
\begin{equation}
E(G_n) =  (S_k-S_k^{-1})Q_n,~~~n=0,1,2,\ldots \ ,
\label{d151}\end{equation}
where $Q_n$ is the characteristic of the symmetry generator $X_n$
introduced above. Since all the necessary quantities that appear are
valued on the horizontal line, we drop the vertical index 
on values of $u$ in all the calculations that follow, and denote the 
horizontal shift simply as $S$ (instead of $S_k$). 

First we verify (\ref{d151}) in the case $n=0$. Using $G_0 = - \ln(u_1 - u_{-1})$ 
and $Q_0=\frac1{u_1-u_{-1}}$ we obtain 
$$ E(G_0) = 
   S^{-1} \frac{\partial G_0}{\partial u_1}
 +  S \frac{\partial G_0}{\partial u_{-1}} 
= - \frac1{u_0-u_{-2}} + \frac1{u_2-u_0} = (S-S^{-1})Q_0 $$ 
as desired. 

Now assume that (\ref{d151}) is true for $n=r-1$. We have 
\begin{eqnarray}
E(G_r)
&=& E(X(G_{r-1})) \nonumber \\
&=& E\left( \sum_{i=-r}^r (S^iQ) \frac{\partial G_{r-1}}{\partial u_i} \right) \nonumber \\ 
&=& \sum_{j=-r-1}^{r+1} \sum_{i=-r}^r 
    S^{-j} \frac{\partial}{\partial u_j} 
   \left[ (S^iQ) \frac{\partial G_{r-1}}{\partial u_i} \right]  \nonumber \\ 
&=& \sum_{j=-r-1}^{r+1} \sum_{i=-r}^r \left[
\left( S^{-j} \frac{\partial(S^iQ) }{\partial u_j} \right) 
\left( S^{-j} \frac{\partial G_{r-1}}{\partial u_i}  \right) 
+ \left(S^{i-j}Q \right) S^{-j} \left( \frac{\partial^2 G_{r-1}}{\partial u_i\partial u_j}  
     \right) \right] \nonumber \\ 
&=& \sum_{j=-r-1}^{r+1} \sum_{i=-r}^r \left[
\left( S^{i-j} \frac{\partial  Q}{\partial u_{j-i}} \right) 
\left( S^{-j} \frac{\partial G_{r-1}}{\partial u_i}  \right) 
+ \left(S^{i-j}Q \right) S^{-j} \left( \frac{\partial^2 G_{r-1}}{\partial u_i\partial u_j}  
     \right) \right] \label{big}
\end{eqnarray} 
Here $Q$ denotes the characteristic of the symmetry $X$. Note that in the second 
term in (\ref{big}), the terms with $j=-r-1$ and $j=r+1$ do not contribute. In fact the 
second term is precisely $X(E(F_{r-1}))$, as the following calculation shows:
\begin{eqnarray*}
X(E(G_{r-1})) 
&=&  X\left( \sum_{j=-r}^r S^{-j}  \frac{\partial G_{r-1}}{\partial u_j}  
   \right)  \\
&=&  \sum_{j=-r}^r \sum_{k=-r-j}^{r-j} (S^k Q) \frac{\partial}{\partial u_k}  
   \left(S^{-j} \frac{\partial G_{r-1}}{\partial u_j}  \right)  \\
&=& \sum_{j=-r}^r \sum_{k=-r-j}^{r-j} (S^k Q) S^{-j} 
     \frac{\partial^2 G_{r-1}}{\partial u_j\partial u_{k+j}} \ .
\end{eqnarray*}
The last expression is seen to be the same as the second term in (\ref{big}) by
replacing the summation variable $k$ by $i=k+j$. With regard to the first term in 
(\ref{big}), note that $Q$ only depends on $u_1$ and $u_{-1}$ and so there 
are only nonzero contributions when $j=i+1$ or $j=i-1$. Thus we have  
\begin{eqnarray*}
E(G_r)
&=& 
   \sum_{i=-r}^r   \left( S^{-1} \frac{\partial  Q}{\partial u_{1}} \right) 
                  \left( S^{-i-1} \frac{\partial G_{r-1}}{\partial u_i}  \right) 
+  \sum_{i=-r}^r   \left( S \frac{\partial  Q}{\partial u_{-1}} \right) 
                  \left( S^{-i+1} \frac{\partial G_{r-1}}{\partial u_i}  \right) \nonumber\\
&& +    X(E(G_{r-1}))   \\
&=& \left(  S^{-1} \frac{\partial  Q}{\partial u_{1}} \right) S^{-1} E(G_{r-1})
+  \left(  S      \frac{\partial  Q}{\partial u_{-1}}\right)   S E(G_{r-1}) 
+    X(E(G_{r-1}))    \\ 
&=&   (S^{-1}-S)\left(    \frac{\partial  Q}{\partial u_{1}} E(G_{r-1})  \right) 
+    X(E(G_{r-1}))  \ .
\end{eqnarray*} 
In the last line we have used the fact that 
$$   \frac{\partial  Q}{\partial u_{-1}}
  = -\frac{\partial  Q}{\partial u_{1}}\  .  $$
Now we use the induction hypothesis $E(G_{r-1}) =  (S-S^{-1})Q_{r-1}$. 
Since shift operators commute with any prolonged symmetry operator we obtain at
once that 
\begin{eqnarray*}
E(G_r)
&=& (S-S^{-1}) \left(  
-\frac{\partial  Q}{\partial u_{1}} (S-S^{-1})Q_{r-1} + X(Q_{r-1}) 
\right)   \\
&=& (S-S^{-1}) \left(  
-(S Q_{r-1}) \frac{\partial  Q}{\partial u_{1}} 
- (S^{-1} Q_{r-1}) \frac{\partial  Q}{\partial u_{-1}} 
+ X(Q_{r-1}) \right)  \\  
&=& (S-S^{-1}) \left(    - X_{r-1}(Q) 
+ X(Q_{r-1}) \right)  
\end{eqnarray*}
But since $X_r=[X,X_{r-1}]$, $Q_r=X(Q_{r-1})-X_{r-1}(Q)$. Thus we have 
$E(G_r)=(S-S^{-1})Q_r$, as required, providing the induction step for 
our claim that $E(G_n)=(S-S^{-1})Q_n$ for all $n$. In particular $E(G_n)\not=0$,
giving nontriviality of the conservations laws with components $F_n,G_n$. 
\end{proof}

\section{The continuum limit} 
Let us consider the continuum limit of (\ref{eq01}).
By replacing
\[u_{i,j}=u(x+ih,t+jh),~~~\alpha=\alpha(h),~~~\beta=\beta(h),\]
and dividing equation (\ref{eq01}) by $h^2$ we obtain
\begin{equation}
\frac{(u(x+h,t+h)-u(x,t))(u(x+h,t)-u(x,t+h))}{h^2}=\frac{\alpha(h)-\beta(h)}{h^2}.\label{eq2}
\end{equation}
Taking the limit of (\ref{eq2}) as $h\rightarrow 0$ we obtain 
\begin{equation}
u_x^2-u_t^2=C,\label{eq4}
\end{equation}
where $C=\lim_{h\rightarrow 0}\frac{\alpha(h)-\beta(h)}{h^2}$, assuming the limit
exists. For brevity we call this limit the continuum limit.
It is clear that the continuum limits of symmetries and conservation laws for 
(\ref{eq01}) give symmetries and conservation laws  for (\ref{eq4}).
For example, the continuum limits of the symmetries in (\ref{SH1}) are
\begin{align}
&X_0=\frac{1}{2u_x}\partial_u,
&\quad&
Y_0=\frac{1}{2u_t}\partial_u,
&\nonumber\\
&X=\frac{x}{2u_x}\partial_{u}+\partial_{C},
&\quad&
Y=\frac{t}{2u_t}\partial_{u}+\partial_{C}.
&\label{d177}
\end{align}
The continuum limits of the conservation laws in (\ref{CLH1}) are
\begin{eqnarray*}
&& F=\ln(u_x+u_t)\ ,~G=-\ln(u_x)\ ,\nonumber \\
&& \bar{F}=\ln(u_t)\ ,~ \bar{G}=-\ln(u_x+u_t)\ ,\nonumber\\
&& \tilde{F}=xF+t\bar{F}\ ,~ \tilde{G}=xG+t\bar{G}\ .\label{CLH1c}
\end{eqnarray*}

The continuum limit of the Rasin-Hydon construction of conservation laws is
as follows:

\begin{theorem}\label{theorem2}
The equation (\ref{eq4}) has an infinite number of distinct, nontrivial conservation laws
generated by repeated application of the symmetry $X$ to the 
conservation law with components $(F,G)$. Writing $F_n=X^n(F)$, $G_n=X^n(G)$ 
we find $G_n=\frac1{u_x^{2n}}$ for $n\ge 1$ (up to addition 
of a trivial conservation law and rescaling). 
\end{theorem}

\begin{proof}
We have 
\begin{eqnarray*}
G_1 ~=~ X(G_0) &=& - \left( \frac{x}{2u_x}\right)_x  \partial_{u_x}  \left( \ln u_x \right) \\
       &=& -  \frac1{2u_x^2}  +    \frac{xu_{xx}}{2u_x^3}  \\
       &=& - \left( \frac{x}{4u_x^2} \right)_x -  \frac1{4u_x^2}\ .
\end{eqnarray*}
The first term on the RHS is the ``$G$'' component of a trivial conservation law,
the second  is a multiple of $\frac1{u_x^2}$,
proving the result for $n=1$. For $n\ge 1$ we have  
\begin{eqnarray*}
X\left(\frac1{u_x^{2n}}\right) 
&=&  \left( \frac{x}{2u_x}\right)_x  \partial_{u_x}  \left(\frac1{u_x^{2n}}\right)  \\
&=&  - \frac{n}{u_x^{2n+2}}  +   \frac{nxu_{xx}}{u_x^{2n+3}}  \\
&=& - \left(  \frac{nx}{(2n+2)u_x^{2n+2}}  \right)_x   - \frac{n(2n+1)}{(2n+2)u_x^{2n+2}}\ . 
\end{eqnarray*}
The first term on the RHS is the ``$G$'' component of a trivial conservation law 
and the second is a multiple of $\frac1{u_x^{2(n+1)}}$.
Thus the required form of $G_n$ is established by induction. 
$G_n$ is evidently not the $x$-derivative
of a function of $u$ and its $x$-derivatives and thus the conservations laws with 
components $F_n,G_n$ are all nontrivial and distinct. 
\end{proof}

Since it is impossible that distinct, nontrivial conservation laws be the 
continuum limit of conservation laws that are equivalent or trivial, 
this furnishes an alternative proof that the conservation
laws constructed by the symmetry method in the discrete case are 
distinct and nontrivial. 

For completeness we give a formula for $F_n$ for $n\ge 1$. From the definition of 
a conservation law for (\ref{eq4}) we have 
\begin{equation}
\frac{\partial F_n}{\partial x}+\frac{\partial G_n}{\partial t}=0\label{dcl1}
\end{equation}
on solutions of (\ref{eq4}). Look for $F_n$ as a function of $u_x$ alone. We then 
need 
$$ F_n'(u_x) u_{xx} - \frac{2nu_{xt}}{u_x^{2n+1}}  = 0 \ .$$ 
But any solution of (\ref{eq4}) also satisfies $u_{x}u_{xx} - u_t u_{tx} = 0$. Thus
$$ F_n'(u_x) = \frac{2nu_{xt}}{u_x^{2n+1}u_{xx}} =  \frac{2n}{u_x^{2n}u_t} 
    =   \frac{2n}{u_x^{2n}\sqrt{u_x^2-C}} \ , $$
and 
$$ F_n(u_x) = \int \frac{2n}{u_x^{2n}\sqrt{u_x^2-C}}\ du_x  \ . $$
These integrals can be computed using $F_1(u_x)=\frac{\sqrt{u_x^2-C}}{Cu_x}$  
and the recursion 
$$ F_n(u_x) = \frac{2n}{(2n-1)C}
\left(
F_{n-1}(u_x) + \frac{\sqrt{u_x^2-C}}{(2n-1)u_x^{2n-1}}
\right)
\ , \quad n>1\ .
$$
This results in an expression for $F_n(u_x)$ 
that is the product of a rational function of $u_x$ and $C$
with $\sqrt{u_x^2-C}$. $\sqrt{u_x^2-C}$ can then be replaced by $u_t$, and, if desired,
all occurences of $C$ can be replaced by $u_x^2-u_t^2$, giving a rational function
of $u_x$ and $u_t$.

\begin{theorem}\label{theorem3}
The continuum limit of the conservation laws constructed using the Gardner transformation
coincides with those constructed by the Rasin-Hydon method.
\end{theorem}

\begin{proof}
In the limit of small $h$, equation (\ref{v1}) gives 
$$ v^{(1)}_{0,0} \sim  - \frac{h}{2u_x} \ . $$
Equation (\ref{v2}) gives 
$$ v^{(i)}_{0,0} \sim  - C_{i-1} \left( \frac{h}{2u_x} \right)^{2i-1} 
  \ , \qquad i=2,3,\ldots \ , $$ 
where $C_n$ are the Catalan numbers, 
$$ C_n = \frac{(2n)!}{n!(n+1)!}\ ,  $$
which satisfy the recursion \cite{wiki1} 
$$ C_0=1 \ , \qquad C_{n+1} = \sum_{i=0}^n C_i C_{n-i}\ , \qquad n\ge 0\ . $$
Using these results in (\ref{G}) gives 
$$ G_n \sim  H_n \left( \frac{h}{2u_x}\right)^{2n} \ , n=1,2,\ldots  $$ 
where the numbers $H_n$ are defined by the generating function 
$$  \sum_{n=1}^\infty H_n t^n    
   = \ln\left( 1 + \sum_{n=1}^\infty C_n t^n  \right) 
   = \ln\left( \frac{1-\sqrt{1-4t}}{2t}  \right) \ ,
$$ 
where in the last equality we have used
the standard expression for the generating function for the 
Catalan numbers \cite{wiki1}. Remarkably, there is a simple 
formula for the $H_n$: 
$$ H_n = \frac{n+1}{2n}C_n = \frac{(2n-1)!}{(n!)^2}  \ , \quad n=1,2,\ldots \ . $$ 
This can be proved from the fact that if we denote the generating function 
of the Catalan numbers by $c(t)=\frac{1-\sqrt{1-4t}}{2t}$, then 
$$ \log(c(t))  = \frac{c(t)-1}{2}   + \int_0^t \frac{c(s)-1}{2s} \ ds  \ . $$
For our purposes, however, it is just necessary to observe that $H_n\not=0$, so, 
to leading order in $h$, $G_n$ is a nonzero multiple of $\frac1{u_x^{2n}}$. 
\end{proof}

To take the limit of the various $F_n$ in (\ref{FsGs}), we use the fact 
that each of the $A_i$ behaves as $-\frac{h}{2u_x}$ while $B$ behaves as
$\frac{h}{u_x+u_t}$. Thus in the limit we obtain $F_n$ as a rational function
of $u_x$ and $u_t$, in agreement with our  previous  conclusions. 

To summarize, in this section we have computed the continuum limit of 
the conservation laws for dKdV found by the Rasin-Hydon method and the 
Gardner transformation method. Since the continuum limits of the conservation
laws are nontrivial, so are the original ones. Furthermore the fact that the limits
of the two sets of conservation laws coincide strongly suggests that the two
sets of conservation laws are identical, but we have not yet succeeded in proving 
this.  

\section{Concluding remarks} 

In this article we have made substantial progress understanding conservation laws for the 
dKdV equation. We have presented two methods for constructing an infinite number of 
nontrivial, distinct conservation laws. The forms of the conservation laws are very similar, 
and we have seen that in a certain continuum limit they coincide, leading us to hypothesize 
that they are in fact equal. 

In addition to proving the two sets of conservation laws are equal, much more remains to be 
done. There are other constructions of the conservation laws for continuum KdV, such as 
the Lenard recursion \cite{Len1} and the method of Drinfeld-Sokolov \cite{DS1}, and it 
is interesting to know if these have analogs for dKdV. We note that proving equivalence of 
the different constructions for continuum KdV is also nontrivial \cite{Wil1}. We mentioned
in section 3 the curious fact that the ``$G$'' components of the first few conservation laws 
constructed by the Gardner method all have the form of a sum of terms of the form $p_kp_{k+1}$.
Sums of the form $\sum_k\psi_k\psi_{k+1}$ also appear as the discrete analog of the $L_2$ norm 
$\int \psi^2 dx$ in the context of the discrete Schr\"odinger equation upon which the 
inverse scattering theory for dKdV is built \cite{It1,It2}. 

Another completely open question is to understand the constraints on the dynamics of the 
dKdV equation that come about as a result of the infinite number of conservation laws. In the
case of continuum KdV, the conservation laws give rise to bounds on Sobolev norms of the 
solution, thus preventing initially smooth data becoming non-smooth \cite{Lax1}. The
conservation laws for dKdV discussed in this paper do not seem to be appropriate for this; 
maybe there are other conservation laws, or maybe some other understanding is appropriate. 

Finally, a lot of the work in this paper can be generalized to other 
equations in the ABS classification, though there are numerous subtleties. 
The homogeneity argument used in the  proof that  the Rasin-Hydon conservation laws are 
nontrivial breaks down, and for certain equations it seems the continuum limits
of the conservation laws are trivial. A paper on this subject is in preparation.

\bibliographystyle{acm} \bibliography{P} 

\end{document}